\documentclass[letterpaper,onecolumn,accepted=2025-08-15]{quantumarticle}
\pdfoutput=1
\usepackage[utf8]{inputenc}
\usepackage[english]{babel}
\usepackage{xspace}\usepackage[colorlinks=true,linkcolor=blue,citecolor=blue]{hyperref}
\usepackage{amssymb,amsmath,amsthm,graphicx,color}
\usepackage{physics}
\usepackage{mathtools}
\usepackage{enumitem}
\usepackage[margin=1in]{geometry}
\usepackage{multirow,array}
\allowdisplaybreaks

\usepackage[capitalise,nameinlink]{cleveref}

\usepackage{algorithm}
\usepackage[noend]{algpseudocode}

\usepackage{bbm}
\theoremstyle{plain}
\newtheorem{theorem}{Theorem}

\newtheorem{corollary}[theorem]{Corollary}

\newtheorem{lemma}[theorem]{Lemma}
\newtheorem{claim}[theorem]{Claim}
\newtheorem{fact}[theorem]{Fact}

\newtheorem{question}[theorem]{Question}

\newtheorem{definition}[theorem]{Definition}

\theoremstyle{plain}

\newcommand{\complexityclass}[1]{{\mathsf{#1}}\xspace}
\renewcommand{\P}{\complexityclass{P}}
\newcommand{\NP}{\complexityclass{NP}}

\newcommand{\BQP}{\complexityclass{BQP}}

\newcommand{\MIP}{\complexityclass{MIP}}

\newcommand{\DQP}{\complexityclass{DQP}}

\newcommand{\QMA}{\complexityclass{QMA}}
\newcommand{\PDQP}{\complexityclass{PDQP}}
\newcommand{\naCQP}{\complexityclass{naCQP}}
\newcommand{\CQP}{\complexityclass{CQP}}

\newcommand{\qpoly}{\complexityclass{qpoly}}
\newcommand{\polycomplexity}{\complexityclass{poly}}
\newcommand{\ALL}{\complexityclass{ALL}}
\newcommand{\SZK}{\complexityclass{SZK}}
\newcommand{\QCMA}{\complexityclass{QCMA}}

\newcommand{\NEXP}{\complexityclass{NEXP}}

\newcommand{\PSPACE}{\complexityclass{PSPACE}}
\newcommand{\PDQMA}{\complexityclass{PDQMA}}

\renewcommand{\Pr}{\mathop{\bf Pr\/}}
\newcommand{\val}{\mathop{\bf val\/}}

\newcommand{\poly}{\mathrm{poly}}

\newcommand{\cyes}{\textnormal{C}_{\textnormal{YES}}}
\newcommand{\pyes}{\textnormal{P}_{\textnormal{YES}}}
\newcommand{\pno}{\textnormal{P}_{\textnormal{NO}}}

\newcommand{\GapCSP}{\textnormal{\textsc{GapCSP}}}
\newcommand{\nul}{\nu_{\text{low}}}
\newcommand{\nuh}{\nu_{\text{high}}}

\newcommand{\Vprotocol}{\mathsf{V}}
\newcommand{\density}{\mathsf{Density}}
\newcommand{\quasichecked}{\mathsf{QuasiCheck}_{\epsilon, \Delta}}
\newcommand{\quasicheck}{\mathsf{QuasiCheck}}

\begin{document}
\title{Superposition detection and QMA with non-collapsing measurements}
 \author{Roozbeh Bassirian}
  \affiliation{University of Chicago}
 \author{Kunal Marwaha}
 \affiliation{University of Chicago}
 \email{kmarw@uchicago.edu}
\maketitle
\begin{abstract}
We prove that $\QMA$ where the verifier may also make a single \emph{non-collapsing} measurement~\cite{aaronson2014space} is equal to $\NEXP$, resolving an open question of Aaronson~\cite{scott_blogpost_pdqp}. We show this is a corollary to a modified proof of $\QMA^+ = \NEXP$~\cite{bassirian2023quantum}. 
At the core of many results inspired by Blier and Tapp~\cite{blier2010quantum} is an unphysical property testing problem 
deciding whether a quantum state is close to an element of a fixed basis.
\end{abstract}
\section{Introduction}
Hidden-variable theories attempt to explain the probabilistic nature of quantum mechanics as downstream effects of parameters beyond the quantum state. If these parameters are accessible to a quantum computer, unusually strong computation is possible. Aaronson and collaborators~\cite{aaronson_dqp,aaronson2014space} proposed 
and studied enhancements to $\BQP$
inspired by these theories.

One such upgrade is the power to make ``non-collapsing'' measurements, exemplified through the class $\CQP$ (Collapse-free Quantum Polynomial-time)~\cite{aaronson2014space}.
After applying each gate or (collapsing) intermediate measurement, a $\CQP$ machine receives a measurement of the current quantum state in the computational basis \emph{without perturbing the state}. \cite{aaronson2014space} proved that such a machine can solve problems in $\SZK$, and in a black-box setting, solve unstructured search faster than Grover's algorithm (but not in polynomial time). This power is dramatically increased with quantum \emph{advice}: Aaronson showed that $\CQP/\qpoly = \ALL$~\cite{aaronson2018pdqpqpoly}, which contrasts with 
$\BQP/\qpoly \subseteq \QMA/\polycomplexity$~\cite{Aaronson_2005_bqpqpoly,aaronson2013characterization}. 

A related concept to quantum advice is quantum \emph{verification} with an untrusted witness; see \cite{Gharibian_2023} for a recent survey. One particularly mysterious version of such a complexity class is $\QMA(2)$~\cite{qma2_defn}: quantum verification where the witness is guaranteed to be \emph{separable}. Even after much effort~\cite{blier2010quantum,aaronson2008power,chen2010short,pereszlényi2012multiprover,harrow2013testing,qma2_yirka,she2022unitary}, little is known about the power of $\QMA(2)$ besides the trivial bounds $\QMA \subseteq \QMA(2) \subseteq \NEXP$.

Recently, \cite{qma2plus} proposed a new approach to show $\QMA(2) = \NEXP$ by introducing the class $\QMA^+(2)$, where the witness is also guaranteed to have \emph{non-negative amplitudes} in the computational basis.
They proved that $\QMA^+(2) = \NEXP$, and suggested the possibility of $\QMA(2) = \QMA^+(2)$.
However, any such result must rely on the separability guarantee of $\QMA(2)$, since it was later shown that $\QMA^+ = \NEXP$~\cite{bassirian2023quantum}.\footnote{$\QMA = \QMA^+ = \NEXP$ would imply, for example, that $\PSPACE = \NEXP$.} In particular, quantum verification with only a \emph{non-negative amplitudes} guarantee is maximally powerful.

Very recently, Aaronson asked the following question:
\begin{question}[\cite{scott_blogpost_pdqp}]
\label{question}
    What is the power of $\QMA$ where the verifier can make non-collapsing measurements?
\end{question}
This computational model is trivially upper-bounded by $\NEXP$, yet Aaronson conjectured it to be \emph{equal} to $\NEXP$. In this work, we show this conjecture to be true, even with a \emph{single} non-collapsing measurement! 
Notably, our proof is inspired by the proof of $\QMA^+ = \NEXP$~\cite{bassirian2023quantum}.
To unify these results, we identify a property testing problem we call \emph{superposition detection}: can one distinguish the set of computational basis states $\mathcal{B}$ from states $\epsilon$-far from $\mathcal{B}$?
The latter states are in \emph{superposition} over computational basis states.
We show that a quantum verifier has the power of $\NEXP$ if it can test this property on arbitrary states --- more specifically, states that are obtained after a collapsing measurement on a quantum witness.

A note for the cautious reader: with a single copy, this property testing task is physically unrealistic --- it even results in distinguishability of statistically equal ensembles  (see \Cref{fact:sup_detection_is_impossible}). 
Any successful detector must access information beyond the quantum state.
We show that a quantum computer equipped with a ``non-collapsing'' measurement can always detect superposition.
On the other hand, $\QMA^+$ only considers a subset of quantum witnesses, all of which admit efficient superposition detection.
As a result, both computational models have the power of $\NEXP$.

With multiple copies, superposition detection is possible --- but generating the copies seems to require \emph{post-selection}~\cite{aaronson2004quantum,kinoshita2018qma2}, inverse exponential precision~\cite{pereszlényi2012multiprover}, or simply too many quantum witnesses~\cite{aaronson2008power,chen2010short,chiesa2011improved}. These are exactly the issues that have plagued approaches inspired by Blier and Tapp~\cite{blier2010quantum} to prove $\QMA(2) = \NEXP$.

\subsection{Techniques}
Our techniques primarily follow that of~\cite{qma2plus,bassirian2023quantum}, which are inspired by the work of Blier and Tapp~\cite{blier2010quantum}. We describe the general approach and then note the differences in this work.

The central challenge is to place an $\NEXP$-complete problem in a variant of $\QMA$. Following~\cite{qma2plus,bassirian2023quantum}, we use a \emph{succinct} constraint satisfaction problem with \emph{constant} gap. In $(1,\delta)$-$\GapCSP$, either all constraints can be satisfied, or at most a $\delta$ fraction of constraints can be satisfied. This problem is $\NEXP$-hard due to the PCP theorem~\cite{AS92,ALMSS98,harsha2004robust}.\footnote{As we do not consider $\NP$-completeness in this work, we often omit the word \emph{succinct} when describing $(1,\delta)$-$\GapCSP$.}

Consider a witness over two registers: the \emph{constraint index} register with a polynomial number of qubits, and a constant-sized \emph{value} register.
The protocol chooses one of two tests to run with some probability. The first test checks the \emph{rigidity} of a quantum state, i.e. accepting if the witness is close to the form $\frac{1}{\sqrt{R}}\sum_{j \in [R]}\ket{j}\ket{\sigma(j)}$. The second test verifies the \emph{constraints} of the CSP by measuring the state in the computational basis and verifying the constraint indexed by the outcome of the first register.
Using several tricks from the PCP literature~\cite{dinur2007pcp,qma2plus}, the second test is able to decide $(1,\delta)$-$\GapCSP$ whenever the witness is \emph{rigid}.

Surprisingly, the \emph{non-negative amplitudes} guarantee is only used in the \emph{first} test, i.e. checking the \emph{rigidity} of the witness. 
With some probability, the verifier encourages a large $\ell_1$-norm by measuring the overlap with the uniform superposition state $\ket{+}$. 
It must now enforce that the witness sends at most \emph{one} basis element in the \emph{value} register per basis element in the \emph{constraint index} register. For example, after measuring the \emph{constraint index} register in the computational basis, the verifier must reject any superposition in the remaining quantum state. \cite{qma2plus, bassirian2023quantum} can enforce this property assuming the witness has non-negative amplitudes, thus deciding the $\NEXP$-complete problem.

In this work, we consider any algorithm that detects \emph{superposition} over constant-sized computational basis states:
\begin{definition}[Superposition detector]
    A \emph{superposition detector} $\mathsf{SupDetect}_{k,\epsilon,\Delta}$ is an algorithm that inputs a quantum state $\ket{\psi}$ on $k$ qubits and outputs $1$ or $0$, such that:
    \begin{itemize}
        \item If $\ket{\psi}$ is a computational basis state on $k$ qubits, then $\mathsf{SupDetect}_{k,\epsilon,\Delta}(\ket{\psi}) = 1$ with probability $1$.
        \item If $|\braket{e}{\psi}|^2 \le 1 - \epsilon$ for every computational basis state $\ket{e}$ on $k$ qubits, then $\mathsf{SupDetect}_{k,\epsilon,\Delta}(\ket{\psi}) = 1$ with probability at most $1-\Delta < 1$.
    \end{itemize}
\end{definition}
We show that if a quantum computer can implement a superposition detector over the \emph{value} register of the quantum witness, then it can check \emph{rigidity} of the witness, and thus decide the $\NEXP$-complete problem. Since the \emph{value} register is constant-sized, we only need a superposition detector where $k$ is some constant.

We use this generalization in two places. 
First, we show that a \emph{single} non-collapsing measurement allows a quantum computer to detect superposition on \emph{any} quantum state, and thus decide any problem in $\NEXP$. This resolves \Cref{question}. 
Curiously, in both the verification and advice setting~\cite{aaronson2018pdqpqpoly}, the power of non-collapsing measurements is magnified when the verifier is given a state it cannot prepare.

Second, we show that if the witness has non-negative amplitudes, a quantum computer can efficiently detect superposition; this reproves the statement $\QMA^+ = \NEXP$~\cite{bassirian2023quantum}.
The \emph{superposition detection} viewpoint gives new intuition for this result:
$\QMA^+$ and $\QMA^+(2)$ \emph{pre-select} quantum witnesses where the verifier can detect superposition.
However, the verifier cannot detect superposition in \emph{all} quantum witnesses unless $\QMA \ne \NEXP$. 
As a consequence,  any proof of $\NEXP \subseteq \QMA(2)$ must handle \emph{some} witnesses in a fundamentally new way.

\subsection{Our results}
We first prove the following statement:
\begin{theorem}[Main theorem]
\label{thm:no_superposition_detectors}
Consider a variant of $\QMA$ where given any constant $k$ and $\epsilon > 0$, 
the verifier is also allowed to partially measure the quantum witness in the computational basis, and then apply $\mathsf{SupDetect}_{k,\epsilon,\Delta}$ for some $\Delta = \Omega(\frac{1}{\poly(n)})$.
Then this variant contains $\NEXP$.
\end{theorem}
This has two direct consequences.
First, since a superposition detector can be efficiently implemented on all quantum states if the verifier can make non-collapsing measurements, we resolve~\Cref{question}:
\begin{corollary}
\label{cor:qma_cqp_equals_nexp}
    Consider the variant of $\QMA$ where the verifier is allowed to make \emph{one} non-collapsing measurement at any point in the computation. Then this variant equals $\NEXP$, even if the measurement outcome can only be used at the end of the computation.
\end{corollary}
Second, since a superposition detector can be implemented for all quantum witnesses with \emph{non-negative amplitudes}, we recover the main result of~\cite{bassirian2023quantum}:
\begin{corollary}[\cite{bassirian2023quantum}]
\label{cor:qmaplus_equals_nexp}
$\QMA^+ = \NEXP$.
\end{corollary}

\subsection{Related work}
\paragraph{$\QMA$ and hidden-variable theories}
Aaronson initiated the study of quantum complexity subject to hidden-variable theories in~\cite{aaronson_dqp}. They introduced the class $\DQP$ (Dynamical Quantum Polynomial-time), and showed evidence that this class was more powerful than $\BQP$.
However, $\DQP$ has subtleties in its definition which make it particularly challenging to study.
To address this, \cite{aaronson2014space} introduced the concept of ``non-collapsing'' measurement through the classes $\CQP$ and $\naCQP$ (non-adaptive $\CQP$).\footnote{$\naCQP$ was called $\PDQP$ (Product Dynamical Quantum Polynomial-time) in a preprint of \cite{aaronson2014space}, and again in subsequent work~\cite{aaronson2018pdqpqpoly,aaronson_pdqma} We use the name $\naCQP$.}
They showed evidence that these classes were only ``slightly more powerful'' than $\BQP$. By contrast, \cite{aaronson2018pdqpqpoly} showed that this power dramatically increased in the presence of quantum advice: $\naCQP/\qpoly = \ALL \ne \BQP/\qpoly$. Aharonov and Regev~\cite{aharonov2003lattice} define a variant of $\QMA$,\footnote{Confusingly, this variant is called $\QMA+$, not to be confused with $\QMA^+$ of \cite{qma2plus}.} which can be interpreted as allowing non-collapsing measurements on the accept/reject qubit at the end of the computation. This variant is equal to $\QMA$; see further discussion in \cite{aaronson2013characterization,scott_blogpost_pdqp}.
Non-collapsing measurements are also connected to ``rewindable'' quantum computation; see~\cite{hiromasa}.

\paragraph{The complexity class $\QMA(2)$}
The power of $\QMA(2)$ is not sensitive to the \emph{number} of separable partitions: \cite{harrow2013testing} showed that $\QMA(2) = \QMA(k)$ for $k$ at most polynomial in $n$ using a primitive called the \emph{product test}. 
One reason to believe that $\QMA(2)$ is powerful is that $\QMA(2)$ can decide $\NP$ using only logarithmic-sized proofs~\cite{blier2010quantum}; by contrast, $\QMA$ with logarithmic-sized proofs is equal to $\BQP$.
However, since the $\QMA(2)$ protocol has a subconstant promise gap, it cannot (so far) be scaled to $\NEXP$~\cite{chen2010short,chiesa2011improved,aaronson2008power,pereszlényi2012multiprover,GallNN12}.
If there exist efficient algorithms that ``disentangle'' quantum states, then $\QMA = \QMA(2)$; \cite{aaronson2008power,akibue2024hardness} show some evidence that no ``disentanglers'' exist. Although $\QMA^+ = \QMA^+(2) = \NEXP$~\cite{qma2plus,bassirian2023quantum}, there exist constants $c,s$ where $\QMA^+_{c,s} \subseteq \QMA$ and  $\QMA^+(2)_{c,s} \subseteq \QMA(2)$. 
In fact, there exists a promise gap for $\QMA^+(2)$ where a sharp transition in complexity occurs (assuming there is a transition)~\cite{jeronimo2024dimension}; it is unknown if the same is true for $\QMA^+$.

\paragraph{Concurrent work}
\cite{aaronson_pdqma} simultaneously and independently study the power of $\QMA$ with non-collapsing measurements. 
They show that this computational model is equal to $\NEXP$, also resolving \Cref{question}. However, their proof uses $O(n \log n)$ non-collapsing measurements, as opposed to just \emph{one} in this work.
Notably, they use different techniques, building on the proofs of $\MIP = \NEXP$~\cite{mip_nexp} and $\naCQP/\qpoly = \ALL$~\cite{aaronson2018pdqpqpoly}. They also explore the power of $\DQP$~\cite{aaronson_dqp} in the contexts of quantum verification and quantum advice. Note that one must make minor adjustments to the definitions of $\CQP$ and $\naCQP$~\cite{aaronson2014space} to allow non-collapsing measurement on a quantum \emph{witness}. \cite{aaronson_pdqma} makes a slightly different choice than this work; we discuss this technicality further in \Cref{sub:qma_cqp_equals_nexp}.

\section{Preliminaries}
We first define $\QMA_{c,s}$. When $c,s$ are not specified, $c-s$ is allowed to be at least any inverse polynomial in input size. In general, the promise gap can be amplified via parallel repetition or with more clever methods~\cite{marriott2005quantum,nagaj2009fast}.
\begin{definition} [$\QMA_{c,s}$]
Let  $c,s: \mathbb{N} \rightarrow \mathbb{R}^+$ be polynomial-time computable functions. A promise problem $\mathcal{L}_{\text{yes}}, \mathcal{L}_{\text{no}} \subseteq \{0, 1\}^*$ is in  $\QMA_{c,s}$ if there exists a $\BQP$ verifier $V$ such that for every $n \in \mathbb{N}$ and every $x \in \{0, 1\}^n$:
\begin{itemize}
\item \textbf{Completeness:} if $x \in \mathcal{L}_{\text{yes}}$, then there exists a state $\ket{\psi}$ on $\poly(n)$ qubits, s.t.\ 
\begin{align*}\Pr[V(x, \ket{\psi}) \text{ accepts}] \geq c(n)\,.\end{align*}
\item \textbf{Soundness:} If $x \in \mathcal{L}_{\text{no}}$, then for every state $\ket{\psi}$ on $\poly(n)$ qubits, we have
\begin{align*}\Pr[V(x, \ket{\psi}) \text{ accepts}] \leq s(n)\,.\end{align*}
\end{itemize}
\end{definition}
We prove \Cref{thm:no_superposition_detectors} using the $\NEXP$-complete problem proposed by~\cite{qma2plus}.
\begin{definition}[CSP system]
    A \emph{$(N, R, q, \Sigma)$-CSP system} $\mathcal{C}$ on $N$ variables with values in $\Sigma$ consists of a set (possibly a multi-set) of $R$ constraints where the arity of each constraint is exactly $q$.
\end{definition}
\begin{definition}[Value of CSP]
    The \emph{value} of a $(N, R, q, \Sigma)$-CSP system $\mathcal{C}$ is the maximum fraction of satisfiable constraints over all possible assignments $\sigma : [N] \to \Sigma$. The value of $\mathcal{C}$ is denoted $\val(\mathcal{C})$.
\end{definition}
\begin{definition}[$\GapCSP$]
    The \emph{$(1, \delta)$-$\GapCSP$ problem} takes as input a $(N,R,q,\Sigma)$-CSP system $\mathcal{C}$. The task is to distinguish whether $\mathcal{C}$ is such that (in completeness) $\val(\mathcal{C}) = 1$ or (in soundness) $\val(\mathcal{C}) \le \delta$.
\end{definition}
\begin{theorem}[\cite{qma2plus}]
\label{thm:nexphardsystem}
There exist constants $q$ and $\delta < 1$ and a family of succinct $(N_i, R_i, q, \Sigma = \{0,1\})$-CSP systems $\{\mathcal{C}_i\}_{i \ge 1}$ for which the $(1, \delta)$-$\GapCSP$ problem is $\NEXP$-complete.
\end{theorem}
In our proof, we use the notion of \emph{rigid} quantum states from \cite{bassirian2023quantum}. These states map each constraint to a single assignment in $\Sigma^q$.
\begin{definition}[Rigid and quasirigid states, \cite{bassirian2023quantum}]
\label{defn:rigid}
Given a $(N, R, q, \Sigma)$-CSP system $\mathcal{C}$, consider
quantum states over a \emph{constraint index} register of dimension $R$ and a \emph{value} register of dimension $\kappa := |\Sigma|^q$, i.e.
\begin{align*}
    \ket{\psi} = \sum_{j \in [R], x \in \Sigma^q} a_{j,x} \ket{j}\ket{x}\,,
\end{align*}
for some $a_{j,x} \in \mathbb{C}$ so that $\braket{\psi}{\psi} = 1$. Then $\ket{\psi}$ is called \emph{quasirigid} if there is a function $\sigma: [R] \to \Sigma^q$ and set of complex numbers $\{b_j\}_{j \in [R]}$ such that $\ket{\psi} =  \sum_{j \in [R]} b_j \ket{j} \ket{\sigma(j)}$, and moreover called \emph{rigid} if all $b_j = \frac{1}{\sqrt{R}}$.
\end{definition}
\cite{bassirian2023quantum} showed a $\QMA$ protocol for $(1,\delta)$-$\GapCSP$, \emph{assuming} the quantum witness is a \emph{rigid} state.
In other words, the protocol is sound \emph{only} for rigid quantum witnesses.
\begin{lemma}[\cite{bassirian2023quantum}]
\label{lemma:rigid_means_qma_nexp}
 There exist constants $0 < \xi \le \cyes < 1$  and a $\QMA_{\cyes, \cyes - \xi}$ protocol named $\mathsf{ConstraintCheck}$ such that for every $(1,\delta)$-$\GapCSP$ instance $\mathcal{I}$ guaranteed by \Cref{thm:nexphardsystem}, $\mathsf{ConstraintCheck}$ decides $\mathcal{I}$ \emph{assuming} the quantum witness is a \emph{rigid} state.
\end{lemma}
\cite{qma2plus,bassirian2023quantum} considered $\QMA^+$, a variant of $\QMA$ such that the quantum witness always has \emph{non-negative amplitudes}. Given this condition, \cite{bassirian2023quantum} designed a protocol to ensure the quantum witness is a \emph{rigid} state, and so $\NEXP \subseteq \QMA^+$. In this work, we design a similar protocol assuming the existence of efficient superposition detectors.
We repeatedly make use of the following fact:
\begin{fact}[\cite{fuchs1998cryptographic}]
\label{fact}
    Let $0 \le \Pi \le \mathbb{I}$ be a positive semi-definite matrix, and let $\ket{\psi_1}$ and $\ket{\psi_2}$ be quantum states such that $|\braket{\psi_1}{\psi_2}|^2 \ge 1 - d$. Then $|\bra{\psi_1}\Pi\ket{\psi_1} - \bra{\psi_2}\Pi\ket{\psi_2}| \le \sqrt{d}$.
\end{fact}
\begin{proof}
    $|\bra{\psi_1}\Pi\ket{\psi_1} - \bra{\psi_2}\Pi\ket{\psi_2}| = |\Tr(\Pi(\ketbra{\psi_1}{\psi_1} - \ketbra{\psi_2}{\psi_2}))|$ is at most the trace distance of $\ketbra{\psi_1}{\psi_1}$ and $\ketbra{\psi_2}{\psi_2}$, which equals $\sqrt{1 - |\braket{\psi_1}{\psi_2}|^2} \le \sqrt{d}$.
\end{proof}

\section{Proof of main theorem}
Our protocol depends on $R, q, \Sigma$ that determine the $(N, R, q, \Sigma)$-CSP system $\mathcal{C}$. It considers quantum witnesses over two registers: a \emph{constraint index} register of dimension $R$ and a \emph{value} register of dimension $\kappa := |\Sigma|^q$, i.e.
\begin{align*}
    \ket{\psi} = \sum_{j \in [R], x \in \Sigma^q} a_{j,x} \ket{j}\ket{x}\,,
\end{align*}
for some $a_{j,x} \in \mathbb{C}$ so that $\braket{\psi}{\psi} = 1$. Before describing the protocol, we define two tests, $\density$ and $\quasicheck$, that act on the quantum witness. We use them in tandem to ensure the witness is \emph{rigid}.
\begin{definition}
    $\density$ takes as input a quantum state, measures it in the Fourier basis, and accepts iff the observed basis element is $\ket{+} := \frac{1}{\sqrt{R \cdot \kappa}} \sum_{j \in [R], x \in \Sigma^q} \ket{j,x}$.
\end{definition}
\begin{definition}
    $\quasichecked$ (with $\epsilon, \Delta > 0$) takes as input a quantum state, measures the \emph{constraint index} register in the computational basis, then sends the remaining quantum state to a superposition detector $\mathsf{SupDetect}_{\log \kappa, \epsilon,\Delta}$, and accepts iff the output of $\mathsf{SupDetect}_{\log \kappa, \epsilon,\Delta}$ is $1$.
\end{definition}
$\density$ is straightforward: it checks the overlap of a quantum state with the uniform superposition state. Let $A(\Vprotocol, \ket{\psi})$ be the acceptance probability of test $\Vprotocol$ when input the quantum state $\ket{\psi}$.
\begin{fact}
\label{fact:density_info}
     $A(\density, \ket{\psi}) = |\braket{+}{\psi}|^2$. Moreover, if $\ket{\psi}$ is \emph{rigid}, then $A(\density, \ket{\psi}) = \frac{1}{\kappa}$.
\end{fact}
$\quasichecked$ collapses the \emph{constraint index} register of the input, and ensures that there is no superposition over the \emph{value} register. A quasirigid state perfectly satisfies this test.
\begin{fact}
\label{fact:quasicheck_info}
If $\ket{\psi}$ is \emph{quasirigid}, then $A(\quasichecked, \ket{\psi}) = 1$. For arbitrary $\ket{\psi}$, 
\begin{align*}
    A(\quasichecked, \ket{\psi}) \le  \sum_{j \in [R]} \left( \sum_{x \in \Sigma^q} |a_{j,x}|^2 \right) \cdot  c_{j,\epsilon}  = (1 - \Delta) + \Delta \cdot\hspace{-4mm} \sum_{j \in [R]; c_{j, \epsilon} = 1} \left( \sum_{x \in \Sigma^q} |a_{j,x}|^2\right) \,,
\end{align*}
where $c_{j,\epsilon} = 1-\Delta$ if $\max_{x \in \Sigma^q} |a_{j,x}|^2 \le (1-\epsilon) \sum_{x \in \Sigma^q} |a_{j,x}|^2$, and $1$ otherwise.
\end{fact}
In a formal sense, $\quasichecked$ ensures that the input is close to a quasirigid state:
\begin{lemma}
\label{lemma:validity_means_close_to_quasirigid}
    Suppose $A(\quasichecked, \ket{\psi}) = w$. Then there exists a \emph{quasirigid} state $\ket{\phi}$ such that $|\braket{\psi}{\phi}|^2 \ge (1-\epsilon)\frac{w-(1-\Delta)}{\Delta}$.
\end{lemma}
\begin{proof}
    Choose a map $f: [R] \to \Sigma^q$ so that for every clause $j \in [R]$, $f(j)$ is an element of $\Sigma^q$ that has the largest weight; i.e. $|a_{j,f(j)}|^2 = \max_{x \in \Sigma^q} |a_{j,x}|^2$. Consider the \emph{quasirigid} state $\ket{\phi} = \frac{1}{\sqrt{\gamma}} \sum_{j \in [R]} a_{j,f(j)} \ket{j}\ket{f(j)}$, where $\gamma = \sum_{j \in [R]}|a_{j,f(j)}|^2$ is chosen so that $\braket{\phi}{\phi} = 1$. Then $|\braket{\psi}{\phi}|^2 = \left| \frac{1}{\sqrt{\gamma}} \sum_{j \in [R]} |a_{j, f(j)}|^2\right|^2 = \gamma$.
    Finally, notice that
    \begin{align*}
        \gamma \ge \sum_{j \in [R]; c_{j,\epsilon} = 1} |a_{j,f(j)}|^2 \ge \sum_{j \in [R]; c_{j,\epsilon} = 1} (1-\epsilon) \sum_{x \in \Sigma^q} |a_{j,x}|^2 \ge (1-\epsilon) \cdot \frac{w - (1-\Delta)}{\Delta}\,.
    \tag*{\qedhere} 
    \end{align*}
\end{proof}
Since quasirigid states only use a $\frac{1}{\kappa}$ fraction of basis elements, they will not test well on $\density$. In fact, no quantum state can do well on both $\density$ and $\quasichecked$:
\begin{lemma}
\label{lemma:quadratic}
    Suppose $A(\density, \ket{\psi}) = w_{\mathsf{D}} \ge \frac{1}{\kappa}$ and $A(\quasichecked, \ket{\psi}) = w_{\mathsf{Q}}$. Then $(w_{\mathsf{D}} -  \frac{1}{\kappa})^2 +  (1-\epsilon) \frac{w_{\mathsf{Q}}-(1-\Delta)}{\Delta} \le  1$.
\end{lemma}
\begin{proof}
By \Cref{lemma:validity_means_close_to_quasirigid}, there exists a quasirigid state $\ket{\phi}$ such that  $|\braket{\psi}{\phi}|^2 \ge (1-\epsilon)\frac{w_{\mathsf{Q}}-(1-\Delta)}{\Delta}$. 
Let $\ket{\phi} = \sum_{j \in [R]} b_j \ket{j}\ket{\sigma(j)}$ for some function $\sigma$ and amplitudes $\{b_j\}_{j \in [R]}$.
Using \Cref{fact}, we see that 
$\left| |\braket{+}{\psi}|^2 - |\braket{+}{\phi}|^2 \right|$ is at most $\sqrt{1 - (1-\epsilon)\frac{w_{\mathsf{Q}}-(1-\Delta)}{\Delta}}$. Note that $|\braket{+}{\psi}|^2 = w_{\mathsf{D}}$, and by Cauchy-Schwarz, $|\braket{+}{\phi}|^2 = \frac{1}{R \cdot \kappa} \left| \sum_{j \in [R]} b_j \right|^2 \le \frac{1}{\kappa} \sum_{j \in [R]} |b_j|^2 = \frac{1}{\kappa}$. Since $w_{\mathsf{D}} \ge \frac{1}{\kappa}$, we have $w_{\mathsf{D}} - \frac{1}{\kappa} \le |\braket{+}{\psi}|^2 - |\braket{+}{\phi}|^2$, which is at most $\sqrt{1 - (1-\epsilon) \frac{w_{\mathsf{Q}}-(1-\Delta)}{\Delta}}$. The statement follows after rearranging terms.
\end{proof}
We show that the \emph{best} quantum states for this pair of tests are close to \emph{rigid} states. Intuitively, the quasirigid states that test well on $\density$ have nearly uniform amplitudes, and so are close to rigid.
\begin{claim}
\label{claim:rigidity}
    Suppose $A(\quasichecked, \ket{\psi}) \ge 1- \Delta \cdot d_\mathsf{Q}$. Then there exists a rigid state $\ket{\chi}$ such that $|\braket{\chi}{\psi}|^2 \ge  \kappa \cdot A(\density, \ket{\psi}) -(\kappa + 1) \sqrt{\epsilon + d_\mathsf{Q}}$. 
\end{claim}
\begin{proof}
    By \Cref{lemma:validity_means_close_to_quasirigid}, there exists a quasirigid state $\ket{\phi}$ such that $|\braket{\psi}{\phi}|^2 \ge (1-\epsilon)(1 - d_\mathsf{Q})$. By \Cref{fact}, we have
    $\left| |\braket{\mu}{\psi}|^2 - |\braket{\mu}{\phi}|^2 \right| \le \sqrt{1 - (1-\epsilon)(1-d_\mathsf{Q})} \le \sqrt{\epsilon + d_\mathsf{Q}}$ for any unit vector $\ket{\mu}$. We use this in two places. First, since  $|\braket{+}{\psi}|^2 = A(\density, \ket{\psi})$, applying \Cref{fact} with $\ket{\mu} = \ket{+}$ implies $|\braket{+}{\phi}|^2 \ge A(\density, \ket{\psi}) -\sqrt{\epsilon + d_\mathsf{Q}}$. Now let $\ket{\chi}$ be the \emph{rigid} state with the same computational basis elements as $\ket{\phi}$. Then
    \begin{align*}
        |\braket{\chi}{\phi}|^2 = \kappa |\braket{+}{\phi}|^2\ge 
        \kappa \cdot A(\density, \ket{\psi}) - \kappa  \sqrt{\epsilon + d_\mathsf{Q}} \,.
    \end{align*}
    Applying \Cref{fact} using $\ket{\mu} = \ket{\chi}$ implies  $|\braket{\chi}{\psi}|^2  \ge  \kappa \cdot A(\density, \ket{\psi}) - \kappa \sqrt{\epsilon + d_\mathsf{Q}} - \sqrt{\epsilon + d_\mathsf{Q}}$.
\end{proof}

We are now ready to prove \Cref{thm:no_superposition_detectors}.
\begin{proof}[Proof of \Cref{thm:no_superposition_detectors}]
Suppose for any constant $k$ and $\epsilon > 0$, there exists some $\Delta$ at least inverse polynomial in input size such that $\mathsf{SupDetect}_{k,\epsilon,\Delta}$ can be efficiently implemented.
We construct the following $\QMA$ verifier that decides every $\GapCSP$ instance guaranteed by \Cref{thm:nexphardsystem}, choosing $p_1, p_2, p_3, \epsilon$ later on:
\begin{itemize}
    \item With probability $p_1$, run $\density$.
    \item With probability $p_2$, run $\quasichecked$.
    \item With probability $p_3$, run the protocol $\mathsf{ConstraintCheck}$ guaranteed by \Cref{lemma:rigid_means_qma_nexp}.
\end{itemize}

\Cref{lemma:rigid_means_qma_nexp} implies the existence of absolute constants $\cyes, \xi,\kappa$. We choose $\epsilon$ and \emph{distance thresholds} $\nul, \nuh$ to be small positive constants that satisfy the following conditions:
\begin{align*}
    \epsilon < \nuh^2 \le \nul^2 \le 1\,, & & \frac{\nuh}{\nul} \le \frac{\xi}{6(1-\cyes)}\,, & & \left( \kappa \nul + (\kappa + 1) \sqrt{\epsilon + \nul}\right)^{1/2} \le \frac{\xi}{2}\,.\end{align*}
For example, this can be done by first making all constants equal and small enough to satisfy the right-most inequality, then reducing $\nuh$ to satisfy the middle inequality, then reducing $\epsilon$ to satisfy the first inequality.

Now we choose the probabilities $p_1, p_2, p_3$. Let
\begin{align*}
    p_1= \frac{1}{Z}\,,& & p_2= \frac{(\nul + \nuh) (1 - \epsilon)}{\Delta (\nuh^2 - \epsilon) Z}\,, & & p_3 = \frac{\nul}{2(1 - \cyes)Z}\,,
\end{align*}
where $Z := 1 + \frac{(\nul + \nuh) (1 - \epsilon)}{\Delta (\nuh^2 - \epsilon)}  + \frac{\nul}{2(1 - \cyes)}$, so that the probabilities sum to $1$.

Given a quantum witness $\ket{\psi}$, the verifier accepts with probability 
\begin{align*}
    p_1 \cdot A(\density, \ket{\psi}) + p_2 \cdot A(\quasichecked, \ket{\psi}) + p_3 \cdot A(\mathsf{ConstraintCheck}, \ket{\psi})\,.
\end{align*}
In completeness, we consider the quantum witness guaranteed by \Cref{lemma:rigid_means_qma_nexp}. Since this state is rigid, by \Cref{fact:density_info} and \Cref{fact:quasicheck_info}, the verifier accepts with probability at least $\pyes := p_1 \cdot \frac{1}{\kappa} + p_2 \cdot 1 + p_3 \cdot \cyes$.

We split soundness into four cases, depending on the acceptance probabilities of $\density$ and $\quasicheck$. We give a graphical picture in \Cref{fig:superposition_qma_nexp}.
\begin{figure}[t]
\centering
\includegraphics[width=0.5\textwidth]{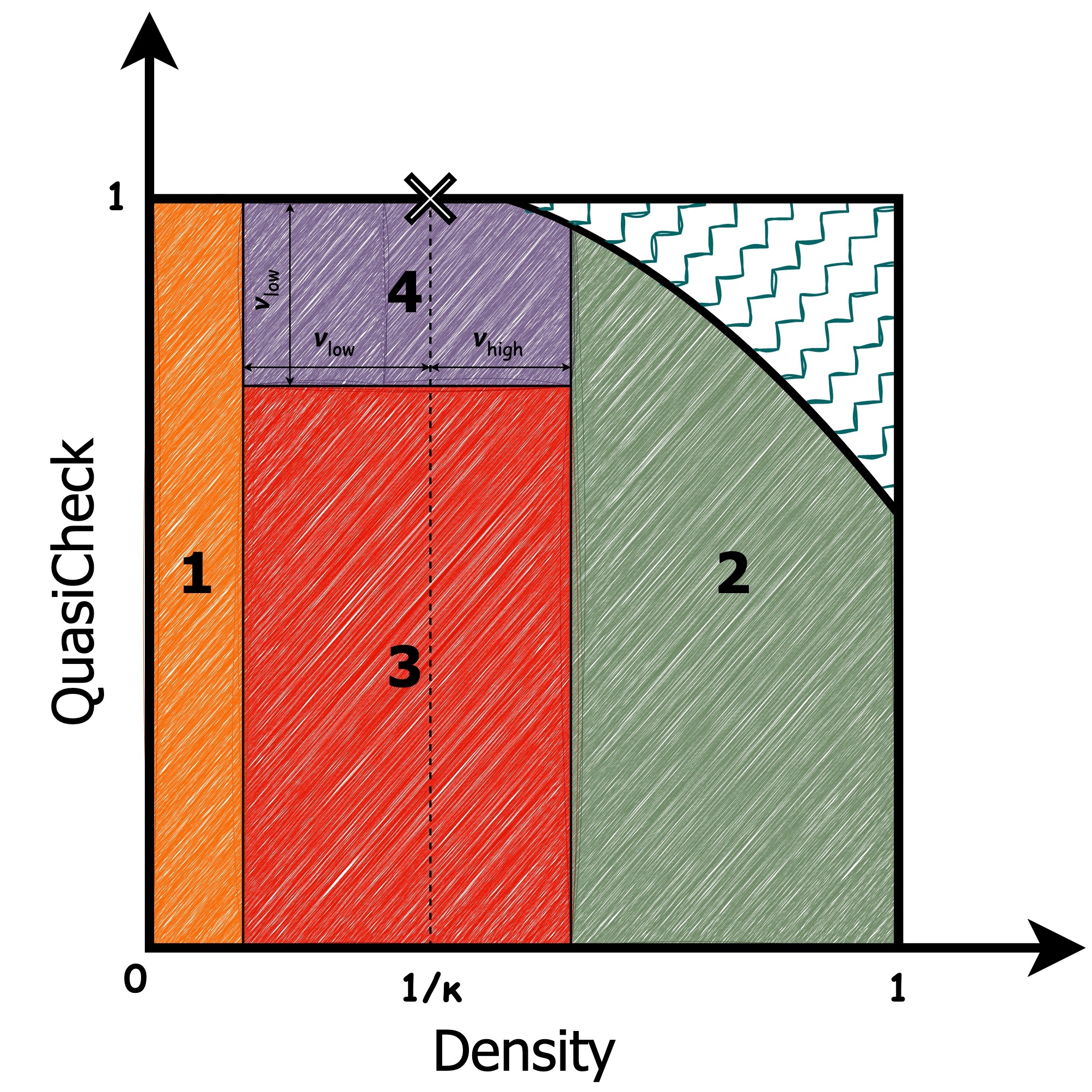}
\caption{Plot of possible acceptance probabilities for the tests $\density$ and $\quasichecked$ when $\Delta = 1$. The numbers $1$ to $4$ correspond to the four cases in soundness in the proof of \Cref{thm:no_superposition_detectors}. The upper right corner is forbidden by \Cref{lemma:quadratic}. Rigid states live at the black ``$\times$'' icon by \Cref{fact:density_info} and \Cref{fact:quasicheck_info}.}
\label{fig:superposition_qma_nexp}
\end{figure}
\begin{enumerate}
    \item In the first case, the witness doesn't test well enough on $\density$. If $A(\density, \ket{\psi}) = \frac{1}{\kappa} - d$ for $d \ge \nul$, then the verifier accepts with probability at most
    \begin{align*}
        \pno &\le p_1 \cdot (\frac{1}{\kappa} - d) + p_2 \cdot 1 + p_3 
        \\
        &\le \pyes - p_1 \cdot \nul + p_3 \cdot (1 - \cyes) 
        \\
        &= \pyes - \frac{\nul}{2 Z}\,.
    \end{align*}
    \item In the second case, the witness tests too well on $\density$, and so cannot test well enough on $\quasicheck$. If $A(\density, \ket{\psi}) =  \frac{1}{\kappa} + d$ for $d \ge \nuh$, then by \Cref{lemma:quadratic}, $A(\quasichecked, \ket{\psi}) \le 1-\Delta + \Delta \cdot \frac{1-d^2}{1 - \epsilon}$:
    \begin{align*}
        \pno &\le p_1 \cdot \left( \frac{1}{\kappa} + d\right) + p_2 \cdot \left(1-\Delta + \Delta \cdot \frac{1-d^2}{1 - \epsilon}\right) + p_3
        \\
        &\le \pyes + p_1 \cdot d - p_2 \Delta \cdot \left(\frac{d^2 - \epsilon}{1 - \epsilon}\right) + p_3 \cdot (1 - \cyes) \,.
    \end{align*}
The right-most expression depends on $d$ through the term 
$$
p_1 \cdot d - p_2\Delta \cdot \left(\frac{d^2 - \epsilon}{1 - \epsilon}\right) = p_1 \left( d - \frac{(\nul + \nuh) (d^2 - \epsilon)}{\nuh^2 - \epsilon} \right)\,.
$$
This term is decreasing with $d$ when $d \ge \nuh$; one can verify this by taking a partial derivative with respect to $d$. So the expression is maximized at $d = \nuh$, and so
\begin{align*}
     \pno &\le  \pyes + p_1 \left( \nuh - \frac{(\nul + \nuh)(\nuh^2 - \epsilon)}{\nuh^2 - \epsilon} \right)+ p_3 \cdot (1 - \cyes) 
     \\
     &=  \pyes - \frac{\nul}{2 Z}\,.
\end{align*}
\item In the third case, the witness doesn't test well enough on $\quasicheck$, and the acceptance probability of $\density$ is not too high. If $A(\density, \ket{\psi}) = \frac{1}{\kappa} + d$ for $-\nul \le d \le \nuh$ and $A(\quasichecked, \ket{\psi}) \le 1-\Delta \cdot \nul$, then
\begin{align*}
    \pno &\le p_1 \cdot \left( \frac{1}{\kappa} + d \right) + p_2 \cdot \left(1-\Delta \cdot \nul\right) + p_3 
    \\
    &\le   \pyes + p_1 \cdot d - p_2 \Delta \cdot \nul +  p_3 \cdot (1 - \cyes)\,.
\end{align*}
By our choice of constants, $d \le \nuh \le \nul$. Moreover, one can verify that $p_2 \Delta \ge 2 \cdot p_1$ whenever $\nuh \le 1$. So $p_1 \cdot d - p_2 \Delta \cdot \nul \le - p_1\cdot  \nul$, and $\pno \le \pyes - p_1 \cdot \nul +  p_3 \cdot (1 - \cyes) =   \pyes - \frac{\nul}{2 Z}$.

\item In the fourth case, the witness tests well on $\quasicheck$, and the acceptance probability of $\density$ is not too low. So the witness must be close to \emph{rigid}, and do worse on $\mathsf{ConstraintCheck}$. If $A(\density, \ket{\psi}) = \frac{1}{\kappa} + d$ for $-\nul \le d \le \nuh$ and $A(\quasichecked, \ket{\psi}) \ge 1-\Delta \cdot\nul$, then by \Cref{claim:rigidity}, there is a rigid state $\ket{\chi}$ such that $|\braket{\chi}{\psi}|^2 \ge 1 + \kappa d - (\kappa + 1) \sqrt{\epsilon + \nul} \ge 1 - \kappa \nul - (\kappa + 1) \sqrt{\epsilon + \nul}$. We use \Cref{fact} with $\Pi$ equal to the accepting projector of $\mathsf{ConstraintCheck}$; then the difference $|A(\mathsf{ConstraintCheck}, \ket{\psi}) - A(\mathsf{ConstraintCheck}, \ket{\chi})| \le \left(\kappa \nul + (\kappa + 1) \sqrt{\epsilon + \nul}\right)^{1/2} \le \frac{\xi}{2}$. Putting this together,
\begin{align*}
     \pno &\le p_1 \cdot \left( \frac{1}{\kappa} + d \right) + p_2 \cdot 1 + p_3 \cdot\left( \cyes - \xi + \frac{\xi}{2}\right)
     \\
     &\le \pyes + p_1 \cdot \nuh - p_3 \cdot \frac{\xi}{2} 
     \\
     &= \pyes + \frac{1}{Z} \left( \nuh - \frac{\nul \cdot \xi}{4(1 - \cyes)} \right)
     \\
     &\le \pyes - \frac{\nuh}{2Z} \,.
\end{align*}
\end{enumerate}
So, this protocol succeeds with completeness $\pyes$ and soundness at most $\pyes - \frac{\nuh}{2Z}$.
\end{proof}
\section{Implications}

\subsection{$\QMA$ with non-collapsing measurements equals $\NEXP$}
\label{sub:qma_cqp_equals_nexp}
The ability to make non-collapsing measurements implies efficient superposition detection. 
\begin{claim}
\label{claim:noncollapse_to_superposition_detector}
    Consider any positive integer $k$ and constant $0 < \epsilon \le 1 - \frac{1}{2^k}$. Then a quantum computer with the ability to make \emph{just one} non-collapsing measurement can efficiently implement $\mathsf{SupDetect}_{k,\epsilon,2(x-x^2)}$, where $x = \min(\epsilon, \frac{1}{2^k})$.
\end{claim}
\begin{proof}
    Suppose the quantum computer is given a $k$-qubit state $\ket{\psi}$, and choose any $\epsilon > 0$. The computer first makes one non-collapsing measurement of $\ket{\psi}$. The computer then makes a (collapsing) measurement of $\ket{\psi}$ in the computational basis, and outputs $1$ if and only if the measurements agree.

    If $\ket{\psi}$ is truly a computational basis state, then the computer always outputs $1$. Now suppose $|\braket{e}{\psi}|^2 \le 1 - \epsilon$ for all $k$-qubit computational basis states $\ket{e} \in \mathcal{B}_k$. Let $\ket{e_{*}}$ be a state in $\mathcal{B}_k$ with the largest squared overlap with $\ket{\psi}$, and let $t = |\braket{e_{*}}{\psi}|^2$. Then the probability of outputting $1$ can be upper-bounded as
    \begin{align*}
        \sum_{\ket{e} \in \mathcal{B}_k} \left( |\braket{e}{\psi}|^2 \right)^2
        =
        |\braket{e_*}{\psi}|^4 + \sum_{\ket{e} \in \mathcal{B}_k; \ket{e} \ne \ket{e_*}}|\braket{e}{\psi}|^4
        \le
        t^2 + \Bigg( \sum_{\ket{e} \in \mathcal{B}_k; \ket{e} \ne \ket{e_*}}|\braket{e}{\psi}|^2\Bigg)^2 = t^2 + (1-t)^2 \,,
    \end{align*}
    which equals $1 - 2(t-t^2)$.
    This is maximized at the largest and smallest valid values of $t$.
    Observe that $t \ge \frac{1}{2^k}$ since $\ket{e_{*}}$ has the largest squared overlap with $\ket{\psi}$, and $t \le 1 - \epsilon$ by supposition.
\end{proof}
The lower bound of \Cref{cor:qma_cqp_equals_nexp} is then a direct consequence of \Cref{thm:no_superposition_detectors} and \Cref{claim:noncollapse_to_superposition_detector}. 
For the upper bound, note that any efficient quantum circuit that includes collapsing or non-collapsing measurements can be classically simulated in exponential time, so $\QMA$ with non-collapsing measurements is at most $\NEXP$. This resolves \Cref{question}.

\cite{aaronson2014space} defined two complexity classes with non-collapsing measurement: $\CQP$ and $\naCQP$ (non-adaptive $\CQP$).
In $\naCQP$, the non-collapsing measurement results can only be used at the \emph{end} of the quantum computation (i.e. in classical post-processing). 
We remark that both \Cref{claim:noncollapse_to_superposition_detector} and the proof of \Cref{thm:no_superposition_detectors} satisfy this non-adaptive condition.

The original framing of \Cref{question} in \cite{scott_blogpost_pdqp} directly considers the power of $\QMA$ with a $\naCQP$ (or $\CQP$) verifier. This has a technical issue which we discuss here. \cite{aaronson2014space} considered a \emph{classical} oracle $\mathcal{O}$ that receives a quantum circuit and outputs non-collapsing measurement results. $\naCQP$ ($\CQP$) is then defined as $\P$ with $1$ ($\poly(n)$) query to $\mathcal{O}$. If we use such a machine as a verifier, it equals $\NEXP$ only if $\mathcal{O}$ can also receive the \emph{quantum} witness. To handle this issue, we consider $\QMA$ where the verifier can make a non-collapsing measurement on an arbitrary quantum state at \emph{any} point in the computation. 
Here, \Cref{claim:noncollapse_to_superposition_detector} uses one non-collapsing measurement, which is immediately followed by one collapsing measurement.

\cite{aaronson_pdqma} instead defines the class $\PDQMA$, where the verifier first runs a quantum circuit (with collapsing measurements), then makes \emph{all} non-collapsing measurements, and finally runs a classical circuit given the non-collapsing measurement results. In this model, \Cref{claim:noncollapse_to_superposition_detector} is not valid as written, since non-collapsing measurements must occur after all collapsing measurements. However, by replacing the final collapsing measurement with a non-collapsing measurement, we recover $\NEXP \subseteq \PDQMA$ using \emph{two} non-collapsing measurements. We remark that this is optimal: $\PDQMA$ with one non-collapsing measurement is equal to $\QMA$, since the verifier's final measurement can always be collapsing. \cite{aaronson_pdqma} uses $O(n \log n)$ non-collapsing measurements to prove $\NEXP \subseteq \PDQMA$.

\subsection{$\QMA^+ = \NEXP$}
\label{sub:qma_plus_equals_nexp} 
If a quantum witness has \emph{non-negative amplitudes}, a partial measurement in the computational basis will preserve this property.
A $\BQP$ machine can implement superposition detection on any such quantum state:
\begin{claim}
\label{claim:qmaplus}
    Consider the algorithm $\mathsf{Verify}^+_k$, which on input a $k$-qubit quantum state $\ket{\psi}$, measures $\ket{\psi}$ in the Fourier basis and outputs $0$ if the observed basis component is $\frac{1}{2^{k/2}}\sum_{i \in [2^k]}\ket{i}$, and $1$ otherwise. Then for every $\epsilon > 0$, there is a superposition detector $\mathsf{SupDetect}_{k, \epsilon, \sqrt{\epsilon/2^k}}$ such that for every state $\ket{\psi}$ with non-negative computational basis amplitudes, $A(\mathsf{Verify}^+_k, \ket{\psi})
    =
    A(\mathsf{SupDetect}_{k, \epsilon, \sqrt{\epsilon/2^k}}, \ket{\psi})
    - \frac{1}{2^k}$.
\end{claim}
\begin{proof}
    Let the input be $\ket{\psi} = \sum_{i \in [2^k]} a_i \ket{i}$. Then $A(\mathsf{Verify}^+_k, \ket{\psi}) = 1 - \frac{1}{2^k} \left| \sum_{i \in [2^k]} a_i\right|^2$. Since all $a_i \ge 0$,
    \begin{align*}
        A(\mathsf{Verify}^+_k, \ket{\psi}) = 1 - \frac{1}{2^k}  \Big(\sum_{i \in [2^k]} a_i\Big)^2 
        = 1 - \frac{1}{2^k}\left( 1 + \sum_{i,j \in [2^k]; i \ne j} a_i a_j\right)
         \le 1 - \frac{1}{2^k}\,,
    \end{align*}
    and this is achieved if $\ket{\psi}$ is a computational basis state. 
    Now, let $i^* = \text{argmax}_{i \in [2^k]} a_i^2$. Since $\sum_{i \in [2^k]} a_i^2 = 1$, $a_{i^*} \ge \frac{1}{2^{k/2}}$. If $a_{i^*}^2 \le 1 - \epsilon$, then
    \begin{align*}
        \sum_{i,j \in [2^k]; i \ne j} a_i a_j 
        \ge  \frac{1}{2^{k/2}} \sum_{j \in [2^k], j \ne i^*} a_j \ge  \frac{1}{2^{k/2}} \sqrt{\sum_{j \in [2^k], j \ne i^*} a_j^2} \ge \frac{\sqrt{\epsilon}}{2^{k/2}}\,.     \tag*{\qedhere} 
    \end{align*}
\end{proof}
\Cref{cor:qmaplus_equals_nexp} is then a direct consequence of \Cref{thm:no_superposition_detectors} and \Cref{claim:qmaplus}. Note that $\mathsf{Verify}^+_k$ acts as a superposition detector, since the additive error $\frac{1}{2^k}$ in \Cref{claim:qmaplus} does not change the promise gap of the protocol in \Cref{thm:no_superposition_detectors}.

In $\QMA^+$, one \emph{only} considers quantum witnesses which admit efficient superposition detection by \Cref{claim:qmaplus}. 
By contrast, \Cref{thm:no_superposition_detectors} implies that this property cannot be tested on \emph{all} quantum witnesses, assuming $\QMA \ne \NEXP$. In fact, it is information-theoretically \emph{impossible}, assuming the density matrix is ``all there is'' in quantum mechanics:
\begin{fact}
\label{fact:sup_detection_is_impossible}
    Fix $k$. Consider the set $S_1$ of $k$-qubit computational basis states and the set $S_2$ of $k$-qubit Fourier basis states. Then a uniformly random choice from $S_1$ has the same density matrix as a uniformly random choice from $S_2$. Despite this, $\mathsf{SupDetect}_{k,\epsilon,\Delta}$ distinguishes these two options for all $\epsilon \le 1-\frac{1}{2^k}$ and $\Delta > 0$.
\end{fact}
Does it help to have more quantum witnesses?
Note that having \emph{two} copies of a quantum state allows one to detect superposition: measure each copy in the computational basis, and output $1$ iff the outcomes agree. However, in \Cref{thm:no_superposition_detectors}, the superposition detector is used on a quantum witness \emph{after a partial measurement}.
In completeness, even with $\poly(n)$ quantum witnesses, there is an inverse exponential chance of receiving two copies of the same post-measurement state. 
So as is, we cannot use the additional quantum witnesses to detect superposition unless we allow an inverse exponential promise gap, or \emph{post-selection}~\cite{aaronson2004quantum}.
We observe that \Cref{thm:no_superposition_detectors} neatly recovers the results $\text{precise}\QMA(2) = \NEXP$~\cite{pereszlényi2012multiprover} and $\text{post}\QMA(2) = \NEXP$~\cite{kinoshita2018qma2}, as well as the $\QMA(2)$ protocol with logarithmic-sized proofs for $\NP$~\cite{blier2010quantum}.

In light of this, we suggest that lower bounds to $\QMA(2)$ must go beyond the $\QMA^+$ framework, as the verifier must be able to handle quantum witnesses without directly detecting superposition.

\section{Open problems}
\begin{enumerate}
    \item Non-collapsing measurement is especially powerful with access to an inefficient quantum state: $\QMA$ with non-collapsing measurements equals $\NEXP$, but $\QCMA$ with non-collapsing measurements is in $\PSPACE$. Can one use this concept to construct a classical oracle separating $\QMA$ and $\QCMA$?
    Similarly, how does the power of $\CQP/\polycomplexity$ compare to the power of $\CQP/\qpoly = \ALL$~\cite{aaronson2018pdqpqpoly}?
    \item The position of the non-collapsing measurement in a quantum computation seems to matter in \Cref{cor:qma_cqp_equals_nexp}: non-collapsing measurement is used \emph{after} a collapsing partial measurement.
    This stymies one approach to prove $\QMA(2) = \NEXP$; see discussion in \Cref{sub:qma_plus_equals_nexp}.
    What is the power of $\BQP$ (or $\QMA$) with the ability to make non-collapsing measurements, but only \emph{before} all collapsing measurements?    
\end{enumerate}

\section*{Acknowledgements}
Thanks to John Bostanci, Bill Fefferman, Sabee Grewal, and Barak Nehoran for comments on a draft of this manuscript.

This work was done in part while visiting the Simons Institute for the Theory of Computing. RB acknowledges support from AFOSR (award number FA9550-21-1-0008). This material is based upon work partially supported by the National Science Foundation under Grant CCF-2044923 (CAREER) and by the U.S. Department of Energy, Office of Science, National Quantum Information Science Research Centers as well as by DOE QuantISED grant DE-SC0020360. KM acknowledges support from the National Science Foundation Graduate Research Fellowship Program under Grant No. DGE-1746045. Any opinions, findings, and conclusions or recommendations expressed in this material are those of the author(s) and do not necessarily reflect the views of the National Science Foundation.
\bibliography{_references}
\bibliographystyle{halpha-abbrv-mod}
\clearpage

\end{document}